\documentclass[runningheads]{llncs}
\usepackage[T1]{fontenc}
\usepackage[backend=biber]{biblatex}
\usepackage{mathtools}
\usepackage[hidelinks,bookmarks,bookmarksopen,bookmarksdepth=2]{hyperref}
\usepackage{cleveref}
\usepackage{graphicx}
\usepackage{svg}
\usepackage[inline]{enumitem}
\usepackage{orcidlink}

\newcommand{\remove}[1]{\ignorespaces}
\newcommand{\mayberemove}[1]{\ignorespaces}

\newcommand{\ftrm}{fail-to-receive model} %
\newcommand{\ftsm}{fail-to-send model} %

\begin{document}

\title{Time Is Not a Healer, but It Sure Makes Hindsight 20:20}

\author{Eli Gafni\inst{1} \and Giuliano Losa\inst{2}\orcidlink{0000-0003-2341-7928}}
\authorrunning{Eli Gafni and Giuliano Losa}

\institute{University of California, Los Angeles, USA\\ \email{eli@ucla.edu}\\
    \and Stellar Development Foundation, San Francisco, USA\\ \email{giuliano@stellar.org}}

\maketitle

\begin{abstract}

In the 1980s, three related impossibility results emerged in the field of distributed computing. First, Fischer, Lynch, and Paterson demonstrated that deterministic consensus is unattainable in an asynchronous message-passing system when a single process may crash-stop. Subsequently, Loui and Abu-Amara showed the infeasibility of achieving consensus in asynchronous shared-memory systems, given the possibility of one crash-stop failure. Lastly, Santoro and Widmayer established the impossibility of consensus in synchronous message-passing systems with a single process per round experiencing send-omission faults.

In this paper, we revisit these seminal results.
First, we observe that all these systems are equivalent in the sense of implementing each other.
Then, we prove the impossibility of consensus in the synchronous system of Santoro and Widmayer, which is the easiest to reason about.
Taking inspiration from Völzer's proof pearl and from the Borowski-Gafni simulation, we obtain a remarkably simple proof.

We believe that a contemporary pedagogical approach to teaching these results should first address the equivalence of the systems before proving the consensus impossibility within the system where the result is most evident.

\end{abstract}

\section{Introduction}

In their famous 1983 paper, Fischer, Lynch, and Paterson~\cite{fischer_impossibility_1983} (hereafter referred to as FLP) established that deterministic consensus is unattainable in an asynchronous message-passing system where one process may fail by stopping.
As a foundational result in distributed computing and one of the most cited works in the field, it is crucial to teach this concept in an accessible manner that highlights the core reason for the impossibility.
However, we believe that the original FLP proof is too technical for this purpose and that its low-level system details can obscure the essence of the proof.

In our quest to simplify the FLP proof, we revisit the subsequent extensions and improvements of the FLP result, including Loui and Abu-Amara's asynchronous shared-memory proof~\cite{loui_memory_1987} and Santoro and Widmayer's impossibility proof for synchronous systems with one process failing to send some of its messages per round~\cite{santoro_time_1989}.
The latter paper was titled "Time is not a healer," which inspired our own title.

While the impossibility of consensus was demonstrated in all of these systems, the proofs did not rely on reductions, but instead rehashed FLP's valency-based argument.
This should have suggested that there are reductions between those models.
In this work, we use elementary simulation algorithms to show that the aforementioned systems can indeed implement each other, and thus it suffices to prove consensus impossible in just one of them.

We then reconsider the impossibility proof in the system that is the easiest to reason about: the synchronous system of Santoro and Widmayer.
In this system, we present a new and remarkably simple proof of the impossibility of consensus, which we believe is of great pedagogical value.

Unlike Santoro and Widmayer, we avoid using a valency argument inspired by FLP.
Instead, we draw ideas from the Borowski-Gafni~\cite{borowsky_generalized_1993} simulation of a 1-resilient system using two wait-free processes and from Völzer's~\cite{volzer_constructive_2004} brilliant impossibility proof.
Völzer's idea, which he used to simplify FLP in the original FLP model, is to compare runs with one missing process with fault-free runs.

Next, we give an overview of the technical contributions of the paper.

\subsection{Four Equivalent Models}

The paper considers four models:

\begin{itemize}
    \item \textbf{The FLP model}. This is the original asynchronous message-passing model of FLP, in which at most one process may crash-stop.
    \item \textbf{The (1-resilient) shared-memory model}. This is an asynchronous shared-memory system in which at most one process may crash-stop.
    \item \textbf{The (1-resilient) fail-to-receive model}. This is a synchronous, round-by-round message-passing system in which processes never crash, but every round, each process might fail to receive one of the messages sent to it.
    \item \textbf{The (1-resilient) fail-to-send model}. This is a synchronous, round-by-round message-passing system in which processes never crash, but every round one process might fail to send some of its messages. This model was originally presented by Santoro and Widmayer~\cite{santoro_time_1989}.
\end{itemize}
For the sake of brevity, in the rest of the paper we usually omit the ``1-resilient'' prefix in the model names.

Assuming we have $n>2$ processes\footnote{$n>2$ is required for the ABD shared-memory simulation algorithm and by the get-core algorithm.}, all the models above solve the same colorless tasks~\footnote{See~\Cref{sec:simulations_and_tasks} for a discussion of colorless tasks.}.
To show this, we proceed in three steps, each time showing that two models simulate each other in the sense of Attiya and Welch~\cite[Chapter 7]{attiya_distributed_2004}.
We write $A \leq B$ when there is an algorithm that simulates $A$ in $B$ (and therefore $B$ is stronger than $A$), and $A \equiv B$ when $A$ and $B$ simulate each other.
The three steps are the following.
\begin{enumerate}
    \item FLP $\equiv$ (1-resilient) shared memory is a well-known result. We can simulate the FLP model in 1-resilient shared memory by implementing message-passing communication using shared-memory buffers that act as mailboxes. In the other direction, we can use the ABD~\cite{attiya_sharing_1995} shared-memory simulation algorithm to simulate single-writer single-reader registers and then apply standard register transformation to obtain multi-reader multi-writer registers~\cite[Chapter 10]{attiya_distributed_2004}. A rigorous treatment of the equivalence between asynchronous message passing and shared memory appears in Lynch's book~\cite[Chapter 17]{lynch_distributed_1996}.
    \item FLP $\equiv$ fail-to-receive. fail-to-receive $\leq$ FLP follows from a simple synchronizer algorithm that is folklore in the field (each process has a current round and waits to receive messages sent in the current round from $n-2$ other processes before moving to the next round).
            In the other direction, we need to guarantee that the messages of all processes except one are eventually delivered; to do so, we simply require processes to keep resending all their messages forever and to do a little bookkeeping do avoid wrongly delivering a message twice.
    \item fail-to-receive $\equiv$ fail-to-send. fail-to-receive $\leq$ fail-to-send is trivial\footnote{This is an instance of the following first-order logic tautology: $\exists y. \forall x .P(x,y)\rightarrow \forall x. \exists y . P(x,y)$}.
        In the other direction, we present in~\Cref{sec:fts-eq-ftr} a simulation based on the get-core algorithm of Gafni~\cite[Chapter~14]{attiya_distributed_2004}.
        Although the simulation relies entirely on this known algorithm, this is a new result.
\end{enumerate}

\subsection{The New Impossibility Proof}

Having shown that all four models above are equivalent, we show the impossibility of deterministic consensus in the model in which it is the easiest: the synchronous model of Santoro and Widmayer.

Inspired by Völzer~\cite{volzer_constructive_2004}, we restrict our attention to fault-free runs and runs in which one process remains silent.
This allows us to inductively construct an infinite execution in which, every round $r$, making a decision depends on one process $p_r$: if $p_r$ fails to send any message, then the decision is $b_r$, but if all messages are successfully sent, then the decision is $\overline{b_r}\neq b_r$.
Both the initial and inductive steps of the construction follow from a straightforward application of the one-dimensional case of Sperner's lemma.

The proof is also constructive in the sense of Constable~\cite{constable_effectively_2011}:
it suggests a sequential procedure that, given a straw-man consensus algorithm, computes an infinite nondeciding execution.

\section{The Models}
\label{sec:models}

We consider a set $\mathcal{P}$ of $n$ deterministic processes.
Each process has a local state consisting of a read-only input register, an internal state, and a write-once output register.
A configuration of the system is a function that maps each process to its local state.
Initially, each input register contains an input value taken from a set of inputs that is specific to the task being solved (e.g. $0$ or $1$ for binary consensus), each process is in its initial internal state, and each process has an empty output register.
An execution is a sequence of configurations starting with in an initial configuration and where each transition from one configuration to the next depends on the model and the algorithm.

Regardless of the model, when a process writes $v$ to its output register, we say that it outputs $v$.
Moreover, we assume that a process never communicates with itself (e.g.\ a process never sends a message to itself).

\subsubsection{The FLP model}
In the FLP model, processes communicate by message passing, and each process takes atomic steps that consist in a) optionally receiving a message, b) updating its local state, and optionally, if it has not done so before, its output register, and c) sending any number of messages.

Processes are asynchronous, meaning that they take steps in an arbitrary order and there is no bound on the number of steps that a process can take while another takes no steps.
However, the FLP model guarantees that every process takes infinitely many steps and that every message sent is eventually received, except that at most one process may at any point fail-stop, after which it permanently stops taking steps and stops receiving messages.

\subsubsection{The 1-resilient shared-memory model}
In the 1-resilient shared-memory model, processes are asynchronous and communicate by atomically reading or writing multi-writer multi-reader shared-memory registers, and at most one process may fail-stop.

\subsubsection{The \ftsm{}}
In the \ftsm{}, processes also communicate by message passing, but execution proceeds in an infinite sequence of synchronous, communication-closed rounds.
Each round, every process first broadcasts a unique message.
Once every process has broadcast its message, each process receives all messages broadcast in the current round, in a fixed order, except that an adversary picks a unique process $p$ and a set of processes $P$ which do not receive the message broadcast by $p$.
Finally, at the end of the round, each process updates its internal state and optionally, if it has not done so before, its output register, both as a function of its current local state and of the set of messages received in the current round before entering the next round.
No process ever fails.

We write $c\xrightarrow{p, P} c'$ to indicate that, starting from configuration $c$ at the beginning of a round, the adversary chooses the process $p$ and drops the messages that $p$ sends to the members of the set of processes $P$, and the round ends in configuration $c'$.
Note that in a transition $c\xrightarrow{p, P} c'$, it is irrelevant whether $p\in P$, since a process does not send a message to itself.
Also note that, because the order in which messages are received is fixed, the triple $c$, $p$, $P$ determines $c'$.

With this notation, an execution is an infinite sequence of the form
\[
c_1 \xrightarrow{p_1, P_1} c_2 \xrightarrow{p_2, P_2} c_3 \xrightarrow{p_3, P_3} \dots
\]
where, in $c_1$, each process has input $0$ or $1$ (if the task is binary consensus), is in the initial internal state, and has an empty output register, and for every $i$, $p_i\in \mathcal{P}$ and $P_i\subseteq \mathcal{P}$.
An example appears in~\Cref{fig:execution}.

\begin{figure}[h]
        \caption{An execution in the \ftsm{}. Each round, the messages of the process selected by the adversary are highlighted with a darker tone.}
        \includegraphics[width=.7\textwidth]{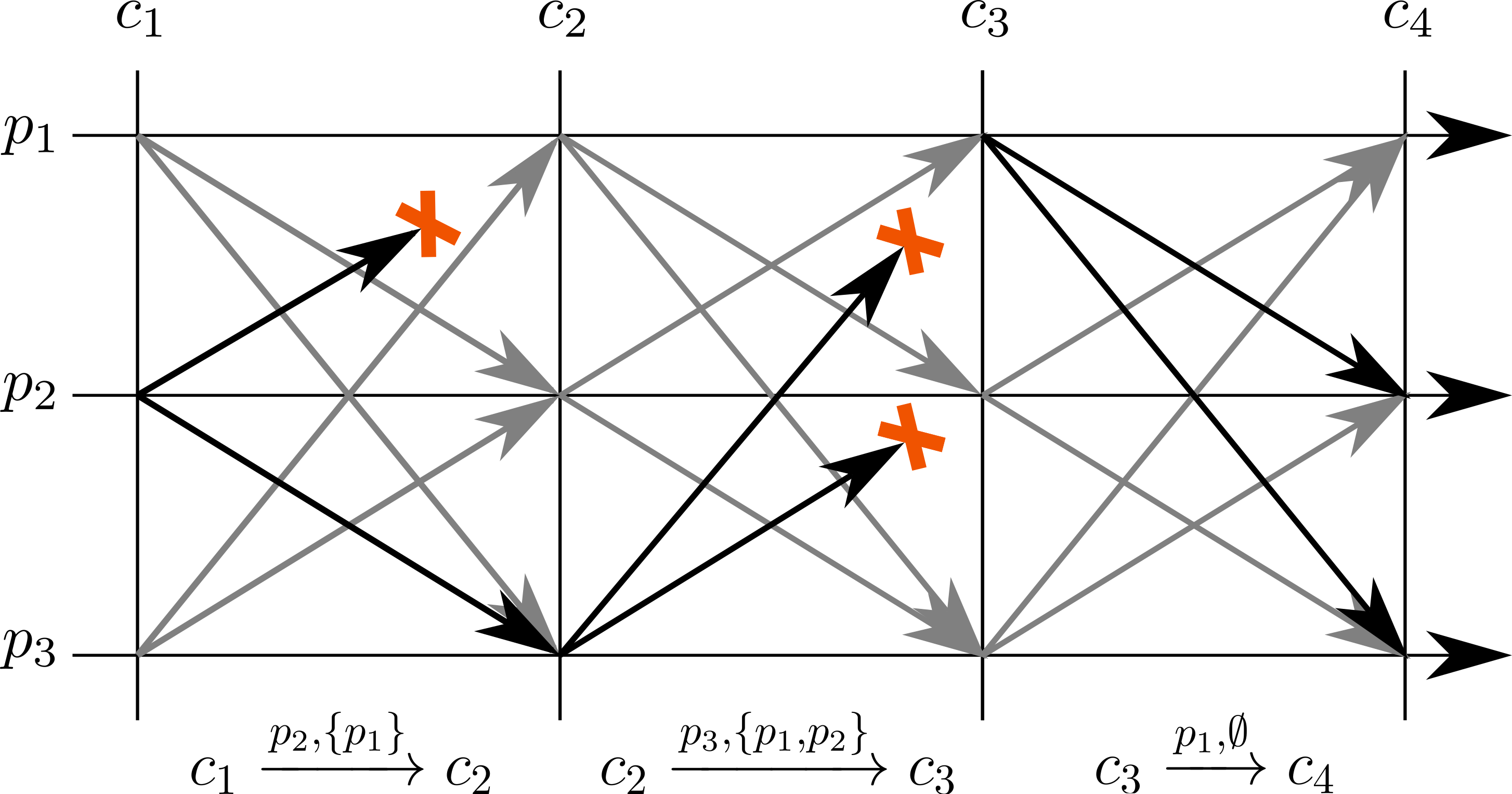}
        \centering
        \label{fig:execution}
\end{figure}

\subsubsection{The \ftrm{}}
The \ftrm{} is like the \ftsm{}, except that the specification of the adversary is different: each round, for every process $p$, the adversary may drop one of the messages addressed to $p$.
So, contrary to the \ftsm{}, it is possible that different processes miss a message from a different process.

\subsection{Simulations and Colorless Tasks}
\label{sec:simulations_and_tasks}

In~\Cref{sec:equiv}, we show using simulation algorithms that the four models above all solve the same colorless tasks.
We now informally define these notions.

A model $B$ simulates a model $A$, written $A\leq B$, when the communication primitives of model $A$, including the constraints placed on them, can be implemented using the communication primitives of model $B$.
For a formal definition of what this means, we point the reader to Attiya and Welch~\cite[Chapter 7]{attiya_distributed_2004} for a formal presentation of simulations.
When models $A$ and $B$ both simulate each other, we write $A\equiv B$.

A colorless task is a relation $\Delta$ between the sets of inputs that the processes may have and the set of outputs that they may produce.
Note that we care only about sets of inputs or outputs, and not about which process has which input or produces which output.
This is what makes a task colorless.

We say that an algorithm in a model solves a colorless tasks when, in every execution:
\begin{enumerate}
    \item Every process that does not fail produces an output, and
    \item If $I$ is the set of inputs held by the processes and $O$ is the set of outputs produced, then $(I,O)\in \Delta$.
\end{enumerate}
For a more precise and rigorous treatment of tasks and colorless tasks, see Herlihy et al.~\cite[Chapter 4]{herlihy_distributed_2013}.

Note that, when solving a colorless task, a process can safely adopt the output of another process.
This is important, e.g., to solve a colorless task $T$ in the \ftrm{} by simulating an algorithm that solves $T$ in the FLP model: Because one process may not output in the FLP model (whereas all processes have to output in the \ftrm{}), this process may need to adopt the output of another.
For a colorless task, this is not a problem.

Informally, we have the following lemma:
\begin{lemma}
    For every two models $A$ and $B$ out of the four models of this section, if $B\leq A$ then $A$ solves all the colorless tasks that $B$ solves.
\end{lemma}

We now define the consensus problem as a colorless task, independently of the model.
\begin{definition}
    In the consensus problem, each process receives $0$ or $1$ as input and the outputs of the processes must satisfy the following properties:
    \begin{itemize}
        \item[Agreement] No two processes output different values.
        \item[Validity] If all processes receive the same input $b$, then no process outputs $b'\neq b$.
    \end{itemize}
\end{definition}

\section{Model Equivalences}
\label{sec:equiv}

In this section, we show that the four models described in~\Cref{sec:models} all solve the same colorless tasks using simulations.
We do not cover the two simulations between the 1-resilient shared-memory model and the FLP model, as this is done brilliantly by Lynch in her book~\cite[Chapter 17]{lynch_distributed_1996}.

\subsection{fail-to-send \texorpdfstring{$\equiv$}{=} fail-to-receive}
\label{sec:fts-eq-ftr}

\subsubsection{fail-to-receive $\leq$ fail-to-send}
The \ftsm{} is a special case of the \ftrm{}: if a single process $p$ fails to send some of its messages (the \ftsm{}), then each process fails to receive from $p$, which is valid in the \ftrm{}.
So the \ftsm{} trivially simulates the \ftrm{}, and we have \ftrm{} $\leq$ \ftsm{}.

\subsubsection{fail-to-send $\leq$ fail-to-receive}
In other direction, it is a-priori not obvious whether we can take a system in which each process may fail to receive from a different process and simulate a system in which all processes may fail to receive from the same process.
Surprisingly, if we have $n>2$ processes, we can simulate the \ftsm{} in the \ftrm{}.

We simulate each round of the \ftsm{} using an instance of the get-core algorithm of Gafni~\cite[Chapter~14]{attiya_distributed_2004}, which takes 3 rounds, which we call phases, of the \ftrm{}.

Each process $p$ starts the first phase with a message that it wants to simulate the sending of and, at the end of the third round, it determines a set of messages to simulate the delivery of.
To obtain a correct simulation, we must ensure that, each simulated round $r$, there is a unique process $p$ such that all processes receive all simulated messages sent in the round except for some messages sent by $p$.

To simulate one round of the \ftsm{}, each process $p$ does the following.
\begin{itemize}
    \item
        In phase~1, $p$ broadcasts its simulated message.
    \item
        In phase~2, $p$ broadcasts the set of simulated messages it received in the first phase and its own simulated message.
    \item
        In phase~3, $p$ broadcasts a message containing the union of all the sets of simulated messages it received in phase~2.
    \item
        Finally, $p$ simulates receiving the union of all the sets of simulated messages it received in phase~3.
\end{itemize}

\begin{lemma}
    \label{lem:emulation}
    When $n>2$, each simulated round, there is a set $S$ of $n-1$ processes such that every process simulates receiving the messages of all members of $S$.
\end{lemma}
\begin{proof}
    Consider a simulated round.
    First, we show that (a) there is a process $p_l$ such that at least $n-2$ processes different from $p_l$ hear from $p_l$ in phase~2.
    Suppose towards a contradiction that, for every process $p$, there are no more than $n-3$ processes that hear from $p$ in phase~2.
    Then, the total number of messages received in phase~2 is at most $n(n-3)$.
    However, in the \ftrm{}, each process receives at least $n-2$ messages (since every process fails to receive at most one message), so at least $n(n-2)$ messages are received by the end of phase~2.
    Since $n(n-2)>n(n-3)$ for $n>2$, this is a contradiction.

    Next, consider the set $S'$ of at least $n-2$ processes different from $p_l$ that $p_l$ hears from in phase~1 and let $S=S'\cup\{p_l\}$.
    Note that $S$ has cardinality $n-1$.
    Let $M_S$ be the set of simulated messages of the members of $S$.

    In phase 2, $p_l$ broadcasts the set of simulated messages received from $S'$ and its own simulated message, i.e.\ $p_l$ broadcasts $M_S$.
    Thus, by Point (a) above, at least $n-1$ processes ($p_l$ and the $n-2$ other processes that receive its message) hold all the simulated messages in $M_S$ by the end of phase~2.
    Thus, because $n-1\geq 2$ for $n>2$, and since the adversary can only prevent each process from receiving one message, all processes receive $M_S$ in phase~3.
    \qed
\end{proof}

\begin{lemma}
    The simulation algorithm is correct.
\end{lemma}
\begin{proof}
    By~\Cref{lem:emulation}, in each simulated round, all processes receive all messages sent in the round except for some messages of a unique process.
    Thus, the simulation algorithm faithfully simulates the \ftsm{}.
\end{proof}

\subsection{FLP \texorpdfstring{$\equiv$}{=} fail-to-receive}
\label{sec:flp-equiv-ftr}

\subsubsection{fail-to-receive $\leq$ FLP}
We simulate the \ftrm{} in the FLP model using a simple synchronizer algorithm.
Each process maintains a current round, initialized to 1, a simulated state, initially its initial state in the simulated model, and a buffer of messages, initially empty.

Each process $p$ obeys the following rules:
\begin{itemize}
    \item
        When $p$ is in round $r$ and $p$ receives a round-$r'$ message, if $r'<r$ then $p$ discards the message and otherwise $p$ buffers the message.
    \item
        When process $p$ enters a round $r$, it broadcasts its round-$r$ simulated message to all.
    \item
        When $p$ is in round $r$ and has $n-2$ round-$r$ simulated messages in its buffer, it simulates receiving all those message and then increments its round number.
\end{itemize}

It is easy to see that if a process $p$ does not fail, then it proceeds from round to round receiving messages sent in the previous round from all other processes except for a single process, which satisfies the constraints of the \ftrm{}.
Thus the simulation is correct.

\subsubsection{FLP $\leq$ fail-to-receive}

The only difficulty in simulating the FLP model in the \ftrm{} is that we have to ensure that every message sent in the simulated algorithm is eventually delivered, except for at most one process, despite the fact that messages can be lost in the \ftrm{}.

To overcome this problem, it suffices that, each round, each process re-broadcast all the simulated messages it has received or sent in previous rounds (by including them in its message for the current round).
In this manner, every simulated message is eventually delivered unless the sender fails to send any messages forever; in this case, the sender can be considered crashed in the simulation.

Finally, to ensure that messages that are sent multiple times in the simulated algorithm can be told apart from messages that are simply re-sent by the simulation algorithm, each process uses a strictly monotonic counter whose value is attached to simulated messages.

\section{Impossibility of Consensus in the Fail-To-Send Model}
\label{sec:ftsm-impossibility}

In this section, we show that consensus is impossible in the \ftsm{}.
To keep the proof constructive, we consider the pseudo-consensus problem, which is solvable, and we show that every pseudo-consensus algorithm has an infinite execution in which no process outputs.
Since solving consensus implies solving pseudo-consensus, this shows that consensus is impossible.

The proof hinges on the notion of $p$-silent execution, which is just an execution in which the adversary drops every message of $p$.
\begin{definition}[p-silent and 1-silent execution]
    We say that an execution $e$ is $p$-silent, for a process $p$, when $e$ is of the form $c_1 \xrightarrow{p,\mathcal{P}}c_2 \xrightarrow{p,\mathcal{P}}c_3 \xrightarrow{p,\mathcal{P}}\dots$.
    We say that an execution is 1-silent when it is $p$-silent for some $p$.
\end{definition}

\begin{definition}[Pseudo-consensus]
    The pseudo-consensus problem relaxes the termination condition of the consensus problem by requiring outputs only in eventually failure-free executions and eventually 1-silent executions.
\end{definition}
Note that pseudo-consensus is solvable, e.g.\ using a variant of the phase-king algorithm~\cite{berman_towards_1989}.
We now consider a pseudo-consensus algorithm.

Throughout the section, we say that a process decides $b$ in an execution when it outputs $b$, and, when a process decides $b$ in an execution, we also say that the execution decides $b$.

\begin{definition}[$p$-dependent configuration]
        A configuration $c$ is $p$-dependent when the decision in the failure-free execution from $c$ is different from the decision in the $p$-silent execution from $c$.
\end{definition}

\begin{lemma}
        \label{lem:p-dep-non-decided}
        If $c$ is a $p$-dependent configuration, then no process has decided in~$c$.
\end{lemma}
\begin{proof}
        Suppose by contradiction that a process has decided on a value $v$ in~$c$.
        Since $c$ is $p$-dependent, there are two executions, starting from $c$, that decide differently.
        Thus, one of these executions decides a value $v'\neq v$.
        This contradicts the agreement property of pseudo-consensus.
\end{proof}

\begin{definition}[Sequence of adjacent configurations]
    We say that a sequence of configurations $c_0,\dots,c_m$ is a sequence of adjacent configurations when, for each $i\in 1..m$, the configurations $c_{i-1}$ and $c_i$ differ only in the local state of a single process noted $p_i$.
\end{definition}

We are now ready to state and prove our main lemma:
\begin{lemma}
        \label{l1}
        Consider a sequence of adjacent configurations $c_0,\dots,c_m$.
        Suppose that the failure-free decision from $c_0$ is different from the failure-free decision from $c_m$.
        Then there exists $k\in 0..m$ and a process $p$ such that $c_k$ is $p$-dependent.
\end{lemma}
\begin{proof}
        Suppose, without loss of generality, that the failure-free decision from $c_0$ is $0$ while the failure-free decision from $c_m$ is $1$.
        Then, there must be $j\in 1..m$ such that the failure-free decision from $c_{j-1}$ is $0$ and the failure-free decision from $c_j$ is $1$ (this is the one-dimensional Sperner lemma).

        We now have two cases.
        First, suppose that the $p_j$-silent decision from $c_j$ is~$0$.
        Then, because the failure-free decision from $c_j$ is 1, we conclude that $c_j$ is $p$-dependent.

        Second, suppose that the $p_j$-silent decision from $c_j$ is $1$.
        Note that, because $c_{j-1}$ and $c_j$ only differ in the local state of $p_j$, if $p_j$ remains silent, then the decision is the same regardless of whether we start from $c_{j-1}$ or $c_j$.
        Thus, the $p_j$-silent decision from $c_{j-1}$ is also $1$.
        We conclude that $c_{j-1}$ is $p_j$-dependent.
        \qed
\end{proof}

We now prove by induction that we can build an infinite execution consisting entirely of $p$-dependent configurations.

\begin{lemma}
    \label{l2}
    There exists a $p$-dependent initial configuration for some process $p$.
\end{lemma}
\begin{proof}
        Order the processes in an arbitrary sequence $p_1,\dots,p_n$.
        Consider the sequence of initial configurations $c_0,\dots,c_n$ where, for each configuration $c_i$ and each process $p_j$, $p_j$ has input $0$ if and only if $j\geq i$.
        Note that the sequence $c_0,\dots,c_n$ is a sequence of adjacent configurations: for each $i\in 1..n$, configurations $c_{i-1}$ and $c_i$ differ only in the input of the process $i$.
        Moreover, by the validity property of consensus, the failure-free decision from $c_0$ is $0$ and the failure-free decision from $c_1$ is $1$.
        Thus, by~\Cref{l1}, there exists a process $p$ and $k\in 0..n$ such that $c_k$ is $p$-dependent.
        \qed
\end{proof}

\begin{lemma}
        \label{l3}
        Suppose that the configuration $c$ is $p$-dependent.
        Then there exists a process $q$, a set of processes $P$, and a configuration $c'$ such that $c \xrightarrow{p,P} c'$ and $c'$ is $q$-dependent.
\end{lemma}
\begin{proof}
        Without loss of generality, assume that the $p$-silent decision from $c$ is $0$.
        Consider the configuration $c'$ such that $c\xrightarrow{p,\mathcal{P}} c'$ (i.e. no process receives $p$'s message in the transition from $c$ to $c'$).
        There are two cases.
        First, suppose that the failure-free decision from $c'$ is $1$.
        Note that the $p$-silent decision from $c'$ must be $0$ because $c'$ is the next configuration after $c$ in the $p$-silent execution from $c$, and $0$ is the $p$-silent decision from $c$.
        Thus, $c'$ is by definition $p$-dependent, and we are done with this case.

        Second, suppose that the failure-free decision from $c'$ is $0$.
        Now order the processes in $\mathcal{P}\setminus\{p\}$ in a sequence $p_1,\dots,p_{n-1}$.
Consider the sequence of configurations $c_1,\dots,c_n$ such that $c\xrightarrow{p,\mathcal{P}}c_1$ (so $c_1=c'$), $c\xrightarrow{p, \mathcal{P}\setminus \{p_1\}}c_2$, $c\xrightarrow{p, \mathcal{P}\setminus \{p_1,p_2\}}c_3$, etc.\ until $c\xrightarrow{p,\emptyset}c_n$. 
        In other words, no process receives the message from $p$ in the transition from $c$ to $c_1$; only $p_1$ receives $p$'s message in the transition from $c$ to $c_2$; only $p_1$ and $p_2$ receive $p$'s message in the transition from $c$ to $c_3$; etc.\ and all processes receive $p$'s message in the transition from $c$ to $c_n$.
        \Cref{fig:lemma} illustrates the situation when there are 3 processors.

        Note that, for each $i\in 1..n-1$, configurations $c_i$ and $c_{i+1}$ only differ in $p_i$ not having or having received $p$'s message; thus, the sequence $c_1,\dots,c_n$ is a sequence of adjacent configurations.
        Moreover, because $c_1=c'$, the failure-free decision from $c_1$ is $0$.
        Additionally, because $c$ is $p$-dependent and the $p$-silent decision from $c$ is $0$, the failure-free decision from $c_n$ is $1$.
        Thus, we can apply~\Cref{l1}, and we conclude that there exists a process $q$ and $i\in 1..n$ such that $c_i$ is $q$-dependent.
        \qed
\end{proof}

\begin{figure}[h]
        \caption{Situation in the second case of the proof of~\Cref{l3}, where $\mathcal{P}=\{p_1,p_2,p_3\}$ (so $n=3$) and $c$ is $p_2$-dependent.
        There must exist $q\in\mathcal{P}$ such that one of the configurations $c_1$, $c_2$, or $c_3$ is $q$-dependent.}
        \includegraphics[width=\textwidth]{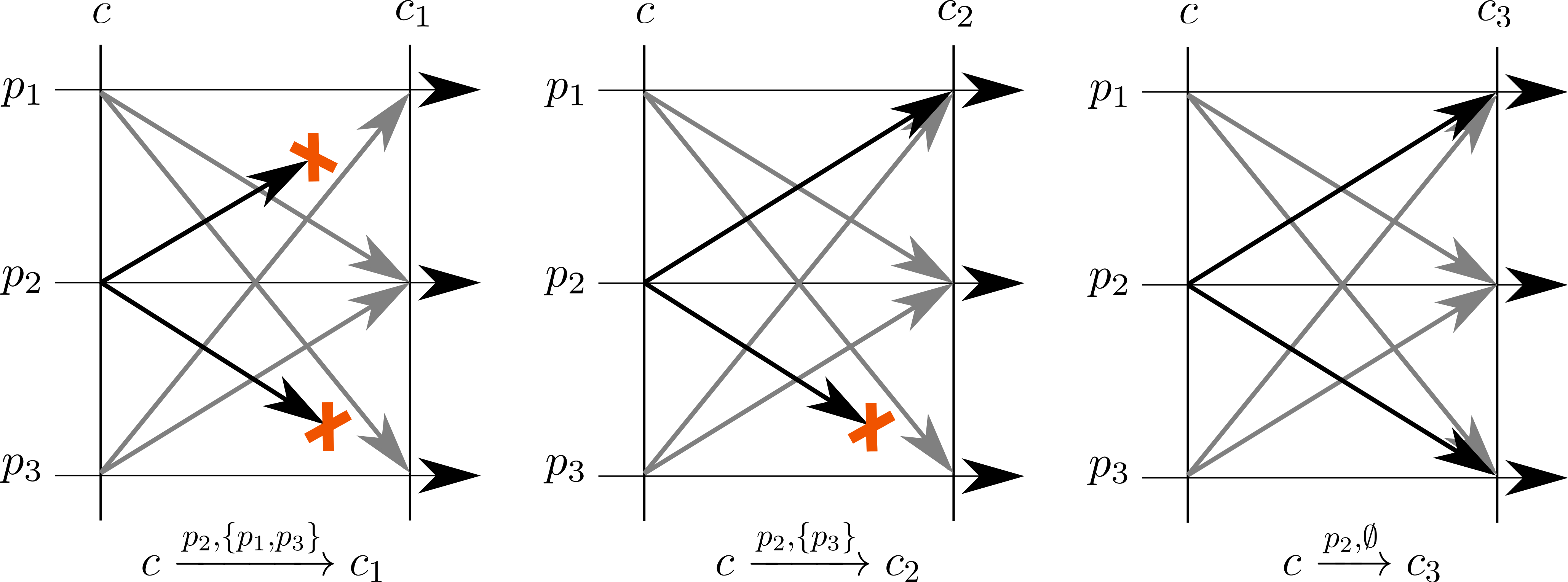}
        \centering
        \label{fig:lemma}
\end{figure}

\begin{theorem}
        Every pseudo-consensus algorithm has an infinite non-deciding execution.
\end{theorem}
\begin{proof}
    Using~\Cref{l2,l3}, we inductively construct an infinite execution in which each configuration is $p$-dependent for some process $p$.
    By~\Cref{lem:p-dep-non-decided}, every $p$-dependent configuration is undecided, and thus no process ever decides.
    \qed
\end{proof}

Note that~\Cref{l3} would fail, and thus the whole proof would fail if, each round, the adversary were constrained to not remove all the messages of the selected process.
This is because we would not be able to construct a sequence of adjacent configurations long enough to go from $c_1$ to $c_n$ (we would be missing one configuration to reach $c_n$ from $c_1$).
In fact, as Santoro and Widmayer remark~\cite{santoro_time_1989}, if the adversary can only remove $n-2$ messages, a protocol proposed in earlier work of theirs solves consensus~\cite{santoro_distributed_1990}.

\section{Related Work}

In 1983, Fischer, Lynch, and Paterson~\cite{fischer_impossibility_1983,fischer_impossibility_1985} first proved the impossibility of solving consensus deterministically in an asynchronous system in which one process may fail-stop.
The proof proceeds by showing that any algorithm that is partially correct (meaning it does not violate agreement and has at least one run that decides 0 and one run that decides 1) has an infinite non-deciding execution consisting of what FLP call bivalent configurations (that is, configurations from which both 0 and 1 can be decided).

Following the FLP result, a number of other works proved similar impossibility results (for deterministic processes) in other models or improved some aspects of the FLP proof.
In 1987, Loui and Abu-Amara\cite{loui_memory_1987} showed that consensus is impossible in shared memory when one process may stop (also proved independently by Herlihy in 1991~\cite{herlihy_wait-free_1991}).

Santoro and Widmayer followed suit in 1989 with the paper ``Time is not a healer''~\cite{santoro_time_1989}, showing, among other results, that, with message-passing communication, even synchrony does not help if, each round, one process may fail to send some of its messages~\cite[Theorem 4.1]{santoro_time_1989}.
\remove{Work from the same authors~\cite{santoro_distributed_1990} presents a consensus algorithm for $n-2$ omission faults, thus showing that $n-1$ is a tight bound.
In the terminology of Santoro and Widmayer, the \ftsm{} of the present paper corresponds to the omission-faults model with $n-1$ faults occurring in a block~\cite[Section 4.1]{santoro_time_1989}.}
The proof of Santoro and Widmayer follows a bivalency argument inspired by the FLP proof.
As we show in~\Cref{sec:equiv}, this result is equivalent to the FLP result and could have been obtained by reduction.

In a pedagogical note, Raynal and Roy~\cite{raynal_note_2005} observe that an asynchronous system with $f$ crash failures and restricted to communication-closed rounds in which each process waits for $n-f$ processes before moving to the next round is equivalent, for task solvability, to the model of Santoro and Widmayer when, each round, each process fails to receive from $f$ processes.

Inspired by Chandy and Misra~\cite{chandy_nonexistence_1985}, Taubenfeld~\cite{taubenfeld_nonexistence_1991} presents a proof of the FLP impossibility in an axiomatic model of sequences of events that avoids giving operational meaning to the events.
This results in a more general and shorter proof.

In their textbook, Attiya and Welch~\cite{attiya_distributed_2004} prove the FLP result by reduction to shared memory.
They first prove that consensus is impossible for two processes and then use a variant of the BG simulation~\cite{borowsky_generalized_1993} to generalize to any number of processes.
Lynch~\cite{lynch_distributed_1996} also takes the shared memory route but proves the shared-memory impossibility using a bivalency argument.

Völzer's proof pearl~\cite{volzer_constructive_2004} gives an elegant, direct proof of the FLP impossibility in the asynchronous message-passing model.
The key insight of Völzer is to build an infinite run consisting of non-uniform configurations, which are bivalent configurations such that different decisions can be reached through $p$-silent executions.
Reading Völzer's paper (in admiration) is the inspiration for the present paper.
Bisping et al.~\cite{bisping_mechanical_2016} present a mechanically-checked formalization of Völzer's proof in Isabelle/HOL.

The latest development regarding the FLP proof, before the present work, is due to Constable~\cite{constable_effectively_2011}.
Constable presents an impossibility proof in the FLP model that roughly follows the FLP proof but that is constructive, meaning that we can extract from this proof an algorithm that, given an effectively non-blocking consensus procedure (in Constable's terminology), computes an infinite non-deciding execution.
The proof of the present paper is also constructive in the same sense.

Finally, the idea of Santoro and Widmayer~\cite{santoro_time_1989} to consider computability questions in a synchronous setting with message-omission faults inspired the development of the general message-adversary model of Afek and Gafni~\cite{afek_simple_2015}; they present a message adversary equivalent to wait-free shared memory and use it to obtain a simple proof of the asynchronous computability theorem~\cite{borowsky_generalized_1993,saks_wait-free_2000,herlihy_topological_1999}.

\printbibliography

\end{document}